\documentclass[11pt]{article}

\newif\ifnotes
\notestrue

\usepackage{hyperref}
\usepackage{fullpage}
\usepackage{amssymb}
\usepackage{amsmath}
\usepackage{amsthm}
\usepackage{amsfonts}
\usepackage{bm}
\usepackage{enumitem}
\usepackage{color}
\usepackage{comment}
\usepackage[capitalize]{cleveref}
\usepackage[dvipsnames]{xcolor}
\usepackage{float}
\usepackage[T1]{fontenc}
\usepackage{todonotes}
\usepackage{asymptote}
\usepackage{mdframed}
\usepackage[most]{tcolorbox}
\usepackage{hyperref}
\usepackage{framed}
\usepackage{mdframed}
\usepackage{scrextend}
\usepackage{bbm}

\usepackage[T1]{fontenc}

\definecolor{asparagus}{rgb}{0.53, 0.66, 0.42}
\definecolor{sacramentostategreen}{rgb}{0.0, 0.3, 0.15}
\definecolor{teal}{rgb}{0.0, 0.5, 0.5}
\definecolor{forestgreen}{rgb}{0.13, 0.6, 0.13}

\usepackage{hyperref}
\hypersetup{
    colorlinks=true,
    linkcolor=sacramentostategreen,
    filecolor=sacramentostategreen,
    citecolor=teal,
    urlcolor=teal,
}
\usepackage[hyperpageref]{backref}

\newtheorem{theorem}{Theorem}[section]
\newtheorem{definition}[theorem]{Definition}
\newtheorem{lemma}[theorem]{Lemma}
\newtheorem{claim}[theorem]{Claim}

\theoremstyle{remark}

\Crefname{theorem}{Theorem}{Theorems}
\Crefname{claim}{Claim}{Claims}
\Crefname{lemma}{Lemma}{Lemmas}
\Crefname{proposition}{Proposition}{Propositions}
\Crefname{corollary}{Corollary}{Corollaries}
\Crefname{definition}{Definition}{Definitions}

\newcommand{\ECC}{\mathsf{ECC}}
\newcommand{\Had}{\mathsf{Had}}

\newcommand{\DEC}{\mathsf{DEC}}

\newcommand{\iECC}{\mathsf{iECC}}
\newcommand{\poly}{\text{poly}}

\newcommand{\Mod}[1]{\ (\text{mod}\ #1)}

\newcommand{\mes}{\mathsf{mes}}

\newcommand{\ind}{\mathsf{ind}}
\newcommand{\cnt}{\mathsf{cnt}}

\newcommand{\parity}{\mathsf{par}}

\newcommand{\nxt}{\mathsf{next}}

\newcommand{\BAD}{\mathsf{BAD}}
\newcommand{\fin}{\mathsf{fin}}

\newcommand{\bbN}{\mathbb{N}}

\newcommand{\bbZ}{\mathbb{Z}}

\makeatletter
\newcommand{\customlabel}[2]{%
   \protected@write \@auxout {}{\string \newlabel {#1}{{#2}{\thepage}{#2}{#1}{}} }%
   \hypertarget{#1}{#2}
}
\makeatother

\newcommand{\protocol}[3]{
    \stepcounter{figure}
    \vspace{0.15cm}
    { \small
    \begin{tcolorbox}[breakable, enhanced, colback=forestgreen!12]
    \begin{center}
    {\bf \underline{Protocol~\customlabel{prot:#2}{\thefigure}: #1}}
    \end{center}
    
    #3
    \end{tcolorbox}
    }
}

\newcounter{casenum}
\newenvironment{caseof}{\setcounter{casenum}{1}}{\vskip.5\baselineskip}
\newcommand{\case}[2]{
    \vskip.5\baselineskip\par\noindent 
    {\bf Case \arabic{casenum}:} {\it #1} \vskip0.1\baselineskip
    \begin{addmargin}[1.5em]{1em}
    #2
    \end{addmargin}
    \addtocounter{casenum}{1}
}

\newcounter{subcasenum}

\newcounter{casenump}

\newcommand{\casep}[2]{
    \ifthenelse{\equal{\value{casenump}}{0}}{
    \vskip.5\baselineskip\par\noindent
    }{}
    {\it Case \arabic{casenump}:} {\it #1}
    \vskip0.1\baselineskip
    \begin{addmargin}[1.5em]{1em}
    #2
    \end{addmargin}
    \addtocounter{casenump}{1}
}

\newcounter{subcasenump}

\begin{document}

\title{Positive Rate Binary Interactive Error Correcting Codes Resilient to $>\frac12$ Adversarial Erasures}
\author{Meghal Gupta \thanks{E-mail:\texttt{meghal@mit.edu}}\\Microsoft Research \and Rachel Yun Zhang\thanks{E-mail:\texttt{rachelyz@mit.edu}. She is supported by an Akamai Presidential Fellowship.}\\MIT}
\date{\today}

\sloppy
\maketitle
\begin{abstract}
An \emph{interactive error correcting code ($\iECC$)} is an interactive protocol with the guarantee that the receiver can correctly determine the sender's message, even in the presence of noise. This generalizes the concept of an error correcting code ($\ECC$), which is a non-interactive $\iECC$ that is known to have erasure resilience capped at $\frac12$. The work of \cite{GuptaTZ21} constructed the first $\iECC$ resilient to $> \frac12$ adversarial erasures. However, their $\iECC$ has communication complexity quadratic in the message size. In our work, we construct the first positive rate $\iECC$ resilient to $> \frac12$ adversarial erasures. For any $\epsilon > 0$, our $\iECC$ is resilient to $\frac6{11} - \epsilon$ adversarial erasures and has size $O_\epsilon(n)$.

\end{abstract}
\thispagestyle{empty}
\newpage
\tableofcontents
\pagenumbering{roman}
\newpage
\pagenumbering{arabic}

\section{Introduction}

Consider the following task: Alice wishes to communicate a message to Bob such that even if a constant fraction of the communicated bits are adversarially tampered with, Bob is still guaranteed to be able to determine her message. This task motivated the prolific study of \emph{error correcting codes}, starting with the seminal works of~\cite{Shannon48,Hamming50}. An error correcting code encodes a message $x$ into a longer codeword $\ECC(x)$, such that the Hamming distance between any two distinct codewords is a constant fraction of the length of the codewords.

An important question in the study of error correcting codes is determining the maximal possible error resilience. It is known that in the adversarial bit-flip model, any $\ECC$ can be resilient to at most $\frac14$ corruptions, and in the adversarial erasure error model any $\ECC$ can be resilient to at most $\frac12$ corruptions. 

This prompts the following natural question: {\em Can we achieve better error resilience if we use interaction?}

In their recent work \cite{GuptaTZ21}, Gupta, Kalai and Zhang introduce the notion of an \emph{interactive error correcting code ($\iECC$)}, which is an interactive protocol with a fixed length and speaking order, such that Bob can correctly learn Alice's input $x$ as long as not too large a fraction of the \emph{total communication} is erased. 
% investigate this question and find that, in fact, interaction can increase the error resilience! They introduce the notion of an \emph{interactive error correcting code ($\iECC$)} (defined formally in Section~\ref{iecc-def}), which is an interactive protocol with a fixed length and speaking order, such that Bob can correctly learn $x$ as long as not too large a fraction of the total communication was erased. 
They demonstrate that $\iECC$'s can in fact achieve a higher erasure resilience than standard error correcting codes. In particular, they design an $\iECC$ that is resilient to adversarial erasure of $\frac35 - \epsilon$ of the total communication. 

Note that a classical error correcting code is an $\iECC$ in which Alice speaks in every round. Their result essentially shows that Bob talking occasionally \emph{instead of Alice} actually improves the error resilience. It is not obvious that this should be the case -- since Bob can only send feedback, while Alice can actually send new information, Bob's messages a priori seem a lot less valuable than Alice's. Nevertheless, they are able to leverage this to improve the erasure resilience past $\frac12$.

However, the size of their protocol is quadratic in the length of Alice's original message $x$. This leaves open the question of whether there exists an $\iECC$ achieving $> \frac12$ erasure resilience with size \emph{linear} in the length of the original message. In this paper, we answer this question to the affirmative.

% Generally, in the theory of error correcting codes, good codes are considered to be ones that are \emph{linear} in the length of the original message, or in other words, codes with a positive rate. In this paper, we present an $\iECC$ that has a positive rate, while still maintaining an erasure resilience $>\frac12$.

\subsection{Our Results}

Our main result is a positive rate $\iECC$ that achieves an erasure resilience of $\frac{6}{11} - \epsilon$ over the binary erasure channel.
\begin{theorem}
    For any $\epsilon > 0$, there exists an $\iECC$ over the binary erasure channel resilient to $\frac6{11}-\epsilon$ erasures, such that the communication complexity for inputs of size $n$ is $O_\epsilon(n)$ and the time complexity is $\poly_\epsilon(n)$.
\end{theorem}

We remark that our $\iECC$ achieves a lower erasure resilience than the quadratic sized $\iECC$ of~\cite{GuptaTZ21}, which is resilient to $\frac35 - \epsilon$ erasures. However, we believe that an $\iECC$ achieving both positive rate and $\frac35 - \epsilon$ erasure resilience can likely be constructed by combining ideas from this paper and~\cite{GuptaTZ21}. Nevertheless, we leave open the existence of such an $\iECC$.

% The erasure resilience our linear protocol achieves of $\frac6{11}$ is slightly worse than the $\frac35$ achieved by the quadratic length protocol of \cite{GuptaTZ21}. In their paper, they also begin with a $\frac6{11}$-resilient protocol and modify it to create a $\frac35$-resilient protocol. With more effort, our new protocol can likely also be modified to achieve $\frac35$ erasure resilience. However, it remains open if any protocol, regardless of length, can achieve error resilience $>\frac35$.

\subsection{Overview of Ideas}

In this overview, we briefly review the $\iECC$ of \cite{GuptaTZ21} then describe how to modify it to have a linear communication complexity. 

The overarching goal of the original protocol, as well as ours, is to perform the following three steps.

\begin{enumerate}
    \item Bob learns that Alice's value of $x$ is one of two possible values. (This idea is known as list decoding, which achieves better noise resilience than unique decoding.) 
    \item Bob conveys to Alice an index $i$ on which the two possible inputs differ.
    \item Alice sends the value of her input at index $i$.
\end{enumerate}

\paragraph{Summary of the Protocol of~\cite{GuptaTZ21}.} The original protocol consists of many (say $\approx \frac{n}{\epsilon}$) \emph{chunks}, where in each chunk Alice sends a message followed by Bob's reply. The protocol is designed so that each such chunk will make \emph{progress} towards Bob's unambiguously learning Alice's input, as long as the adversary did not invest more than $\frac35 - \epsilon$ erasures in that chunk. At a high level, in the first chunk with $< \frac35 - \epsilon$ erasures, Bob narrows down Alice's input to at most two options. In every future chunk with $<\frac35 - \epsilon$ erasures, either Alice gets closer to learning the index $i$ on which the two options differ, or Bob fully determines $x$ by ruling out one of the two values of $x$, e.g. by learning the value of $x[i]$ or by uniquely decoding Alice's message. Alice keeps track of a counter $\cnt$ initially set to $0$ indicating her guess for $i$. The main purpose of Bob's messages is to increment Alice's counter to $i$. 

At the beginning of the protocol, Alice sends $\ECC(x, \cnt)$ to Bob in every chunk. At the first point there are $<\frac35 - \epsilon$ erasures in a chunk, Bob will be able to list decode Alice's message to at most two options, say $(x_0, \cnt_0 = 0)$ and $(x_1, \cnt_1 = 0)$. This must happen because the relative message lengths of Alice and Bob will be such that the adversary cannot corrupt too much of Alice's message even if they corrupt none of Bob's message. Since we are in the setting of erasures, one of the two decodings must be Alice's true state, and in particular must contain Alice's true input.

At this point, Bob begins signaling to Alice to increment $\cnt$. His goal is to tell Alice to increment $\cnt$ until $\cnt = i$. He does this by only sending one of two codewords\footnote{In our protocol, Bob will send one of four codewords each message. This contributes to the lower erasure resilience of $\frac6{11} - \epsilon$.} every message that have relative distance $1$ apart. 
This way, if Bob's message is not entirely erased, Alice learns what Bob tried to send. The key is that every time $< \frac35 - \epsilon$ of a chunk is corrupted,
%He uses only 4 codewords $\bar{0},\bar{1},\bar{2},\bar{3}$ that are distance $\frac23$ rather than $\frac12$ apart to accomplish this over many rounds, making progress every time $<\frac6{11}$ of a chunk is corrupted. The key is that that every time $<\frac6{11}$ is corrupted, by choosing the relative message lengths correctly, 
we can guarantee \emph{both that Bob will decode Alice's message to two possible messages, and Alice uniquely decodes Bob's message},\footnote{It is also possible that instead Bob uniquely decodes Alice's message, but then he will have uniquely learned $x$.} so that Alice and Bob make progress towards Alice learning $i$. 
% Therefore Alice is guaranteed to hear his message, but it is not important for our purposes how exactly this contributes to Alice learning $i$ over time.
Once Alice has discovered $i$, Bob signals for Alice to send the bit $x[i]$ for the rest of the protocol,\footnote{The reader familiar with the $\iECC$ of~\cite{GuptaTZ21} may recall that in the case that the two Alices Bob sees have different values of $\cnt$, Bob may instruct Alice to send a different bit for the rest of the protocol, but we do not address this for now.} which allows him to distinguish whether Alice has $x_0$ or $x_1$.

\paragraph{Modifications to Achieve Positive Rate.}

The communication complexity of the above protocol is $O(n^2)$. This comes from two parts: (1) $O(n)$ chunks are necessarily for Bob to communicate the index $i \in [n]$ to Alice via incrementation, and (2) Alice sends her length $n$ input in every chunk. We show how to lessen both requirements, thus making the final protocol linear in length.

First, for Bob to communicate $i \in [n]$ to Alice, instead of incrementing $\cnt$ by $1$ until it equals $i$, which requires $O(n)$ rounds of interaction, he builds $i$ bit by bit. That is, Bob writes $i$ out in binary, and then sends Alice each bit of this binary representation in sequence. This only requires $O(\log n)$ rounds of interaction. 

Second, we show that instead of sending $x$ every message, it suffices for Alice to encode a shorter string that is different than the corresponding short string for any other $x'$ in \emph{most} chunks. More precisely, consider an error correcting code $\ECC$ with the following property: set some $\alpha$ and for any $x \not= x'$, 
\[
    \ECC(x)[j \alpha, (j+1)\alpha - 1] \not= \ECC(x')[j\alpha, (j+1)\alpha - 1]
\]
for all but an $\epsilon$ fraction of values $j$. Then, Alice rotates through the sections, sending $\ECC(x)[j\alpha, (j+1)\alpha - 1]$ in the $\left( j~\text{mod}~\frac{|\ECC(x)|}{\alpha} \right)$'th chunk. Then, if Bob has narrowed down Alice's input to $x_0$ and $x_1$, he can simply ignore the $\epsilon$ fraction of chunks in which $\ECC(x_0)[j\alpha, (j+1)\alpha - 1] = \ECC(x_1)[j\alpha, (j+1)\alpha - 1]$. In the remainder of chunks, the segment $\ECC(x)[j\alpha, (j+1)\alpha - 1]$ is sufficient for Bob to distinguish between $x_0$ and $x_1$. If we were to let $\alpha\approx \log{n}$,\footnote{$\alpha = \Theta(\log n)$ is necessary since Alice also sends her current guess of $i$ each message, which has length $\log n$.} then our chunks are now only length $O(\log{n})$.

Combining the two modifications, we see that $\Theta(n)$ communication from Alice is necessary for Bob to narrow down Alice's input to two options, and then after that, Bob can convey $i$ to Alie in $O(\log n)$ chunks each of size $O(\log n)$. This results in an $\iECC$ with total communication $O(n + \log^2 n) = O(n)$.

% Combining the two modifications, we obtain a protocol with $O(\log{n})$ chunks, each of size $O(\log n)$, for total communication $O((\log{n})^2)$. Recall however that the protocol cannot be shorter than $O(n)$ since Bob also needs to initially list-decode Alice's message, thus the total communication is $O(n)$.

We remark that our protocol has erasure resilience $\frac6{11} - \epsilon$. The limiting factor is in the construction of a protocol in which Bob builds $i$ bit by bit: our protocol requires Bob sending $4$ codewords with distance $\frac23$, which could possibly be improved to $2$ codewords with distance $\frac12$ to achieve $\frac35 - \epsilon$ erasure resilience, though we do not do that here. However, combining our second observation with the protocol from~\cite{GuptaTZ21} would be enough to give an $\iECC$ with $\frac35 - \epsilon$ erasure resilience with communication $O(n \log n)$. 

% A more detailed explanation of these modifications as well as the formal protocol are included in Section~\ref{sec:protocol}.
\section{Preliminaries and Definitions}
Before we dive into the technical part of our paper, we present important preliminaries on classical error correcting codes, and define an $\iECC$ formally and what it means for one to be resilient to $\alpha$-fraction of erasures.

\paragraph{Notation.} In this work, we use the following notations.
\begin{itemize}
    \item The function $\Delta(x, y)$ represents the Hamming distance between $x$ and $y$.
    \item The interval $[a,b]$ for $a,b\in \bbZ_{\ge 0}$ denotes the integers from $a$ to $b$ inclusive. The interval $[n]$ denotes the integers $1, \dots, n$.
    \item The symbol $\perp$ in a message represents the erasure symbol that a party might receive in the erasure model.
    \item When we say Bob $k$-decodes a message, we mean that he list decodes it to exactly $k$ possible messages Alice could have sent in the valid message space.
    \item The output of an $\ECC$ is $0$-indexed. All other strings are $1$-indexed.
\end{itemize}

\subsection{Classical Error Correcting Codes}

\begin{definition} [Error Correcting Code]
    An error correcting code ($\ECC$) is a family of maps $\ECC = \{ \ECC_n : \{ 0, 1 \}^n \rightarrow \{ 0, 1 \}^{m(n)} \}_{n \in \bbN}$. An $\ECC$ has \emph{relative distance $\alpha > 0$} if for all $n \in \bbN$ and any $x \not= y \in \{ 0, 1 \}^n$,
    \[
        \Delta \left( \ECC_n(x), \ECC_n(y) \right) \ge \alpha m(n).
    \]
\end{definition}

Binary error correcting codes with relative distance $\approx \frac12$ are well known to exist with linear blowup in communication complexity.

\begin{theorem}[\cite{GuruswamiS00}] \label{thm:ECC-plain}
    For all $\epsilon > 0$, there exists an explicit linear error correcting code $\ECC_\epsilon = \{ \ECC_{\epsilon, n} : \{ 0, 1 \}^n \rightarrow \{ 0, 1 \}^{m} \}_{n \in \bbN}$ with relative distance $\frac12 - \epsilon$ and with $m = m(n) = O_\epsilon (n)$. Furthermore, all codewords other than $\ECC_{\epsilon, n}(0^n)$ are relative distance $\frac12 - \epsilon$ from $0^m$ and $1^m$ as well.
\end{theorem}

A relative distance of $\frac12$ is in fact optimal in the sense that as the number of codewords $N$ approaches $\infty$, the maximal possible relative distance between $N$ codewords approaches $\frac12$. We remark, however, that for small values of $N$, the distance can be much larger: for $N = 2$, the relative distance between codewords can be as large as $1$, e.g. the codewords $0^M$ and $1^M$, and for $N = 4$, the relative distance can be as large as $\frac23$, e.g. the codewords $(000)^M, (110)^M, (101)^M, (011)^M$. Our constructions leverage this fact that codes with higher relative distance exist for a small constant number of codewords.

We will also need the following important lemma about the number of shared bits between any three codewords in an error correcting code scheme that has distance $\frac12$. 

\begin{lemma} \label{lem:1/4}
    For any error correcting code $\ECC_\epsilon = \{ \ECC_{\epsilon, n} : \{ 0, 1 \}^n \rightarrow \{ 0, 1 \}^{m} \}_{n \in \bbN}$ with relative distance $\frac12 - \epsilon$, and any large enough $n \in \bbN$, any three codewords in $\ECC_{\epsilon, n}$ overlap on at most $\left( \frac14 + \frac32 \epsilon \right) \cdot m$ locations.
\end{lemma}

\begin{proof}
    Consider three codewords $c_1, c_2, c_3$. Suppose that all pairs are relative distance at least $\left( \frac12 - \epsilon \right)$. Let $c_1$ and $c_2$ share $f \le \left( \frac12 + \epsilon \right) \cdot m$ bits, and all three codewords share $e \le f$ bits. Then, note that 
    \begin{align*}
        2m \cdot \left( \frac12 - \epsilon \right)
        &\le \Delta(c_1, c_3) + \Delta(c_3, c_2) \\ 
        &\le 2(f-e) + (m-f) \\
        &= m + f - 2e \\
        &\le \left( \frac32 + \epsilon \right) \cdot m - 2e,
    \end{align*}
    which means that
    \[
        e \le \left( \frac14 + \frac32 \epsilon \right) \cdot m,
    \]
    as claimed.
\end{proof}

Lemma~\ref{lem:1/4} means that assuming that $< \frac34$ of a codeword is erased, the resulting message is list-decodable to a set of size $\le 2$, at least in theory. The following theorem says such a code exists with list-decoding being polynomial time, while also satisfying a couple other properties necessary in the protocol construction in Section~\ref{sec:protocol}.

\begin{theorem} \cite{Guruswami03} \label{thm:linear}
    For all $\epsilon>0$, any explicit (given with its encoding matrix) linear code $\ECC_\epsilon = \{ \ECC_{\epsilon, n} : \{ 0, 1 \}^n \rightarrow \{ 0, 1 \}^{m} \}_{n \in \bbN}$ with relative distance $(\frac12-\epsilon)$, can be efficiently decoded and list-decoded. That is, there exists a $\poly_\epsilon(n)$-time decoding algorithm $\DEC_\epsilon = \{ \DEC_{\epsilon,n} : \{ 0, 1 \}^m \rightarrow \mathcal{P}(\{ 0, 1 \}^n) \}_{n \in \bbN}$, such that for any $n \in \bbN$, $x \in \{ 0, 1 \}^n$, and corruption $\sigma$ consisting of fewer than $(\frac12 - \epsilon) \cdot m$ erasures,
    \[
        x = \DEC_{\epsilon,n}(\sigma \circ \ECC_{\epsilon, n}(x)).
    \]
    Moreover, for any corruption $\sigma$ consisting of fewer than $(\frac34 - \frac32\epsilon) \cdot m$ erasures,
    \[
        \left| \DEC_{\epsilon,n}(\sigma \circ \ECC_{\epsilon, n}(x)) \right| \le 2, \qquad x \in \DEC_{\epsilon,n}(\sigma \circ \ECC_{\epsilon, n}(x)).
    \]
\end{theorem}

Our following theorem gives an $\ECC$ such that any two codewords differ on most segments of length $\alpha$.

\begin{theorem} \label{thm:ECC}
    For all $n \in \bbN, \epsilon > 0$, there exists $\alpha = \Theta_\epsilon(\log{n})$ such that, there exists an explicit linear code $\ECC_\epsilon = \{ \ECC_{\epsilon, n} : \{ 0, 1 \}^n \rightarrow \{ 0, 1 \}^{m} \}_{n \in \bbN}$ with $m = m(n) = O_\epsilon (n)$, satisfying the following property:
    For all $n \in \bbN, \epsilon > 0$, it holds that $\alpha | m$ and for any $x \not= x' \in \{ 0, 1 \}^n$ and $j\in \{0\dots \frac{m}{\alpha}-1\}$, 
        \[
            \ECC_{\epsilon, n}(x)[ j\alpha : (j+1)\alpha - 1 ] = \ECC_{\epsilon, n}(x')[ j\alpha : (j+1)\alpha - 1 ] \label{prop:jlogn}
        \]
        for at most $\frac{\epsilon m}{\alpha}$ values of $j$.
\end{theorem} 

\begin{proof}
    The paper of \cite{GuruswamiS00} and many others provide an explicit concatenated linear code where the outer code is a large alphabet ($|\Sigma|=O_\epsilon(1)$) linear code $\ECC'_{\epsilon, n} : \{ 0, 1 \}^n \rightarrow \Sigma^{m'}$ of distance $1-\epsilon$ and the inner code is a distance $\frac12$ Hadamard code $\Had : \Sigma \rightarrow \{ 0, 1 \}^r$, so that $m = m'r$. 
    % We modify this code slightly to satisfy the property we desire while still keeping it a concatenated linear code. 
    
    We choose $\alpha'=\lfloor\frac{\epsilon}{2}\log{n}\rfloor$ and assume that $\alpha' | m'$ (this can be done by possibly padding the outer code with up to $\alpha'$ $0$'s; note that the code is still linear after padding with $0$'s). 
    % This decreases the relative distance of the outer code to 
    % \[
    %     \ge 1-\frac{\epsilon}{2}-\frac{\alpha'}{n}\geq 1-\frac{\epsilon}{2}-\frac{\lfloor\frac{\epsilon}{2}\log{n}\rfloor}{n} \geq 1-\epsilon
    % \]
    % since $n>\log{n}$.
    Now, let $\alpha=\alpha'r$. Clearly, $\alpha | m$. To see why~\eqref{prop:jlogn} holds, notice that in order for some $j \in [0, \frac{m}{\alpha} - 1]$ to satisfy
    \[
        \ECC_{\epsilon, n}(x)[j\alpha:(j+1)\alpha-1] = \ECC_{\epsilon, n}(x')[j\alpha:(j+1)\alpha-1]
    \]
    it must be the case that
    \[
        \ECC'_{\epsilon, n}(x) \left[ j\alpha', (j+1)\alpha' - 1 \right]
        = \ECC'_{\epsilon, n}(x') \left[ j\alpha', (j+1)\alpha' - 1 \right].
    \]
    Let $J$ be the number of such $j \in [0, \frac{m}{\alpha} - 1]$. Then since $\ECC'_{\epsilon, n}$ has distance $\ge 1 - \epsilon$ it holds that $J \le \frac{\epsilon m}{\alpha}$, as claimed.
\end{proof}

\begin{lemma} \label{lemma:chunk-code}
    Let $\ECC'_\epsilon = \{ \ECC'_{n, \epsilon} : \{ 0, 1 \}^n \rightarrow \{ 0, 1 \}^{m(n)} \}$ be a explicit linear code satisfying the properties of Theorem~\ref{thm:ECC} with $\alpha = \alpha_n = \Theta_\epsilon(\log n)$. For all linear $\ECC_\epsilon = \{ \ECC_{n, \epsilon} : \{ 0, 1 \}^{\alpha} \times \{ 0, 1 \}^{\beta} \rightarrow \{ 0, 1 \}^{p(n)} \}$ with relative distance $\frac12-\epsilon$ and for all $\beta$, the code defined by 
    \[
        C(x) = \ECC_\epsilon(\ECC'_\epsilon(x)[0,\alpha-1], 0^\beta) || \dots || \ECC_\epsilon(\ECC'_\epsilon(x)[m-\alpha, m-1], 0^\beta)
    \]
    is a linear code with relative distance $\frac12-\frac32\epsilon$. In particular, assuming that less than $\frac34 - \frac94 \epsilon$ of $C(x)$ is erased, there is an efficient algorithm to obtain a set of size $2$ containing $x$.
\end{lemma}

\begin{proof}
    Regardless of the choice of $\ECC$, since both $\ECC$ and $\ECC'$ are linear, it follows that $C$ is linear as well.
    
    We now show that the relative distance between $C(x)$ and $C(x')$ is at least $\frac12-\frac32\epsilon$. By Lemma~\ref{thm:ECC}, at most $\epsilon$ chunks $\ECC'[j\alpha,(j+1)\alpha-1]$ are identical for $x$ and $x'$. For the remaining $1-\epsilon$ fraction of the chunks, the relative distance is at least $\frac12-\epsilon$, so the total relative distance is at least $$\left(\frac12-\epsilon\right)(1-\epsilon)>\frac12-\frac32\epsilon.$$
    
    % \rnote{delete} Finally, we find $\ECC$ such that each $\ECC(\ECC'(x)[j\alpha,(j+1)\alpha-1],i)$ chunk for each $x,i,j$ consists of is relative distance $\frac12-\epsilon$ from $0^*$ and $1^*$. It is well known that there exists a linear code with relative distance $\frac12-\epsilon$ from $1^*$. Now, notice that $\ECC$ is a function that takes in a $0-1$ string, but $(\ECC'(x)[j\alpha,(j+1)\alpha-1],i)$ is a pair. We separate the two elements of the pair with a $1$ (substituting for the comma). For any $\ECC$, we have $\ECC(0)=0^*$, and since none of the inputs $(\ECC'(x)[j\alpha,(j+1)\alpha-1],i)$ map to $0$, the evaluation of $\ECC$ must be at least distance $\frac12-\epsilon$ from $0^*$ as well.
\end{proof}

\subsection{Interactive Error Correcting Codes} \label{iecc-def}
We formally define our notion of an \emph{interactive error correcting code ($\iECC$)}. The two types of corruptions we will be interested in are erasures and bit flips. We first start by defining a non-adaptive interactive protocol.

\begin{definition} [Non-Adaptive Interactive Protocol]
    A non-adaptive interactive protocol $\pi = \{ \pi_n \}_{n \in \bbN}$ is an interactive protocol between Alice and Bob, where in each round a single party sends a single bit to the other party. The order of speaking, as well as the number of rounds in the protocol, is fixed beforehand. The number of rounds is denoted $|\pi|$. 
\end{definition}

\begin{definition} [Interactive Error Correcting Code]
    An interactive error correcting code ($\iECC$) is a non-adaptive interactive protocol $\pi = \{ \pi_n \}_{n \in \bbN}$, with the following syntax: 
    \begin{itemize}
        \item At the beginning of the protocol, Alice receives as private input some $x \in \{ 0, 1 \}^n$.
        \item At the end of the protocol, Bob outputs some $\hat{x} \in \{ 0, 1 \}^n$.
    \end{itemize}
    We say that $\pi$ is $\alpha$-resilient to adversarial bit flips (resp. erasures) if there exists $n_0 \in \bbN$ such that for all $n > n_0$ and $x \in \{ 0, 1 \}^n$, and for all online adversarial attacks consisting of flipping (resp. erasing) at most $\alpha \cdot |\pi|$ of the total communication, Bob outputs $x$ at the end of the protocol with probability $1$.
\end{definition}
\section{$6/11$ Protocol} \label{sec:protocol}

\subsection{Overview}

Let $\ECC' : \{ 0, 1 \}^n \rightarrow \{ 0, 1 \}^m$ be an error correcting code satisfying the statement of Theorem~\ref{thm:ECC} with $\alpha = \Theta(\log n)$, and let $\ECC : \{ 0, 1 \}^\alpha \times \{ 0, 1 \}^{\le \log n} \rightarrow \{ 0, 1 \}^p$ be an error correcting code with distance $\frac12$ that is also relative distance $\frac12$ from $0^p, 1^p$. 

Our $\iECC$ consists of $O_\epsilon \left( \frac{m}{\alpha} \right)$ chunks, each consisting of Alice sending a $p$-bit message followed by Bob sending a $\frac{3p}8$-bit message. Bob's messages are always one of four words $\bar{0}, \bar{1}, \bar{2}, \bar{3} \in \{ 0, 1 \}^{3p/8}$ with relative distance $\frac23$. We outline our protocol below. In what follows, we assume that all messages Bob receives are consistent with the same two values of $x$, otherwise Bob can rule out one of the values of $x$ and determine Alice's true input.

\begin{enumerate}
    \item Alice initially holds a string $\ind \in \{ 0, 1 \}^{\leq\log n}$ initially set to the empty string $\ind = \emptyset$. Alice begins the protocol by sending $\ECC(\ECC'(x)[j\alpha, (j+1)\alpha-1], \ind)$ to Bob in every chunk.
    
    \item Bob begins the protocol sending $\bar{0}$ every message. Every $\frac{m}{\alpha}$ chunks, he attempts to list-decode Alice's previous $\frac{m}{\alpha}$ messages to find consistent values of $x$. Note that by Lemma~\ref{lemma:chunk-code}, if there are at most $\frac34 - \frac32 \epsilon$ erasures in Alice's message in those $\frac{m}{\alpha}$ chunks, then Bob is guaranteed to find at most two possible values of $x$.
    
    \item When Bob has found two consistent values of $x$, say $\hat{x}_0$ and $\hat{x}_1$, he determines an index $i \in [n] = \{ 0, 1 \}^{\log n}$ such that $\hat{x}_0[i] \not= \hat{x}_1[i]$. His goal is now to communicate $i$ to Alice, bit by bit. He does this by sending either $\bar{0}$, $\bar{1}$, or $\bar{2}$ every chunk. 
    
    To communicate the $\nxt$'th bit of $i$, Bob adds $i[\nxt] + 1$ to $\mes$ modulo $3$, where $\overline{\mes}$ was the last message he sent, to get his new message $\mes'$, and begins sending $\overline{\mes'}$ every chunk. (When Alice receives a message from Bob that is different from the last message she received, she can calculate the difference in the two messages to determine the bit.) He does this until he list-decodes Alice's message to two possibilities $\ECC(\ECC'(\hat{x}_0)[j\alpha, (j+1)\alpha-1, \ind_0)$ and $\ECC(\ECC'(\hat{x}_1)[j\alpha, (j+1)\alpha-1], \ind_1)$ such that $\ECC'(\hat{x}_0)[j\alpha, (j+1)\alpha-1] \not= \ECC'(\hat{x}_1)[j\alpha, (j+1)\alpha-1]$, where at least one of $\ind_0$, $\ind_1$ has length $\nxt$. If both have length $\nxt$, he proceeds to communicate the $(\nxt+1)$'th bit of $i$ in the same way. If only one of the two Alice's has $|\ind_b| = \nxt$, Bob switches to sending $\bar{3}$ for the rest of the protocol, signaling to Alice to send him the parity of $|\ind|$ so that he can distinguish between whether Alice has $(\hat{x}_0, \ind_0)$ or $(\hat{x}_1, \ind_1)$.
    
    \item Whenever Alice unambiguously sees a \emph{change} in Bob's message from a $\bar{0}$ to a $\bar{1}$ or $\bar{2}$ (or cyclic), she calculates $b = \mes' - \mes - 1~\text{mod}~3$ and appends $b$ to $\ind$. If she ever receives a $\bar{3}$, she switches to sending $(|\ind|~\text{mod}~2)^p$ for the rest of the protocol. Otherwise, at some point she has $|\ind| = \log n$, so she can convert $\ind$ into an index $i \in [n]$ and send $(x[i])^p$ for the rest of the protocol. Note that Bob can distinguish between $\hat{x}_0$ and $\hat{x}_1$ using the value of $x$ at the index $i$.
\end{enumerate}

In the above outline, one has to be careful around $|\ind| = \log n-1$. In particular, if one Alice has $|\ind| = \log n - 1$ and the other has $|\ind| = \log n$, the second will be sending $x[i]$ for the rest of the protocol and it is thus incorrect for Bob to send $\bar{3}$ to signal the first Alice to send the parity of the length of $\ind$. Instead, once Bob has list-decoded Alice's message such that $|\ind_0| = |\ind_1| = \log n - 1$, Bob commits to sending the next message $\in \{ \bar{0}, \bar{1}, \bar{2} \}$ that conveys to Alice the final bit of $i$ for the rest of the protocol. 

\subsection{Protocol} \label{sec:6/11-protocol}

\protocol{Interactive Binary One Way Protocol Resilient to $\frac6{11} - O(\epsilon)$ Erasures}{6/11}{
    Let $n$ be the size of the message $x \in \{ 0, 1 \}^n$ that Alice wishes to convey to Bob. Let
    \[ 
        \ECC' : \{ 0, 1 \}^n \rightarrow \{ 0, 1 \}^m
    \] 
    be an error correcting code satisfying the statement of Theorem~\ref{thm:ECC} with $\alpha = O_\epsilon(\log n)$, and let 
    \[ 
        \ECC(\cdot, \cdot) : \{ 0, 1 \}^\alpha \times \{ 0, 1 \}^{\le \log n}  \rightarrow \{ 0, 1 \}^p
    \]
    be a code with relative distance $\frac12 - \epsilon$ that is also relative distance $\frac12 - \epsilon$ from $0^p$ and $1^p$, such that $m = O_\epsilon(n)$ and $p = O_\epsilon(\log n)$.
    
    In what follows, $\bar{0}, \bar{1}, \bar{2}, \bar{3}$ denote length $\frac{3p}{8}$ binary strings that have relative distance $\frac23$ from each other (specifically, $(000)^{p/8},(011)^{p/8},(101)^{p/8},(110)^{p/8}$). Our protocol consists of $T = \frac{m}{\epsilon \alpha}$ chunks of Alice sending a $p$ bit message followed by Bob sending a $\frac{3p}{8}$ bit message. Throughout the protocol, the parties will choose a response based on which case applies to their received message $M$; if multiple apply, they choose the first case on the list.
    
    \begin{center}
    {\bf \fbox{Alice}}
    \end{center}
    
    In addition to $x$, Alice has an internal state consisting of
    \begin{itemize}
        \item A string $\ind \in \{ 0, 1 \}^{\le n}$ that at the beginning of the protocol is set to the empty string $\emptyset$.
        \item An internal state $\mes \in \{ 0, 1, 2, 3 \}$, originally set to $0$, representing Bob's most recent message that successfully got through to Alice.
    \end{itemize}
    
    Number the chunks from $1$ to $T$. In each chunk, let $j$ be the residue modulo $\frac{m}{\alpha}$ of the chunk's number minus $1$. Every chunk until otherwise instructed, Alice always sends $\ECC(\ECC'(x)[j\alpha, (j+1)\alpha-1], \ind)$ (recall that the output of $\ECC'$ is $0$-indexed).
    
    Each time after receiving a message $M \in \{ 0, 1, \perp \}^{3p/8}$ from Bob, she updates her internal state, or switches to sending a single bit for the rest of the protocol, as follows. Note that if $< \frac23$ of the symbols of $M$ are erasures, there is only one value of $\gamma \in \{ 0, 1, 2, 3 \}$ such that Bob could've sent $\bar{\gamma}$.
    
    \begin{caseof}
        \case {$\geq \frac23$ of the symbols in $M$ are $\bot$.}{
            Alice makes no update to $\ind, \mes$.
        }
        
        \case{$M$ uniquely decodes to $\gamma \in \{ 0, 1, 2 \}$.}{
            If $\gamma = \mes$, Alice makes no update to $\ind, \mes$.
            
            Otherwise, let $b = (\gamma - \mes - 1 ~\text{mod}~ 3) \in \{ 0, 1 \}$. Alice sets $\ind \gets \ind || b$ and $\mes \gets \gamma$.
            
            If $|\ind| = \log n$, Alice interprets $\ind$ as an index $i \in [0, n-1]$ and sends $x[i]^{p}$ every message for the rest of the protocol.
        }
        
        \case{$M$ uniquely decodes to $\gamma = 3$.}{
            Alice sends $(|\ind|~\text{mod}~2)^p$ for the rest of the protocol.
        }
    
    \end{caseof}
    
    \begin{center}
    {\bf \fbox{Bob}}
    \end{center}
    
    Bob holds a variable $\hat{x}$, initially set to $\emptyset$, that will be updated with his final output either at the end of the protocol or once he has unambiguously learns Alice's value of $x$. Once $\hat{x}$ is set to a value ($\neq \emptyset$), it will not be updated again. That is, Bob ignores any future instructions to update it. At the end of the protocol, Bob outputs $\hat{x}$. If at any point in the protocol $\hat{x}$ has already been set, Bob may send Alice any arbitrary message, say $\bar{1}$. 
    
    Bob also keeps track of the following values:
    \begin{itemize}
        \item Two values $\hat{x}_0$ and $\hat{x}_1$, to be set when Bob determines two possible values of Alice's input for the first time. 
        \item A fixed index $i$ where $\hat{x}_0$ and $\hat{x}_1$ differ, set as soon as $\hat{x}_0$ and $\hat{x}_1$ are known. Bob stores $i$ as a $\log n$-bit binary string.
        \item A set $\BAD$ containing all $j \in [0, \frac{m}{\alpha}-1]$ for which $\ECC'(\hat{x}_0)[j\alpha, (j+1)\alpha-1] = \ECC'(\hat{x}_1)[j\alpha, (j+1)\alpha-1]$, set as soon as $\hat{x}_0, \hat{x}_1$ are determined. 
        \item $\mes \in \{ 0, 1, 2 \}$, representing the last message $\overline{\mes}$ he sent Alice; $\mes$ is originally set to $0$.
        \item $\nxt \in [1, \log n]$, originally set to $1$, denoting the next bit of $i$ that Bob is trying to send to Alice.
        \item $\fin \in \{ 0, 1, 2, 3 \}$ which will Bob will set when he transitions to Phase 2 to denote the message he sends for the rest of the protocol.
        \item $\parity \in \{ 0, 1 \}$ which Bob will set at the end of Phase 1 to use in Phase 2 to distinguish between $\hat{x}_0$ and $\hat{x}_1$.
    \end{itemize}

    Each chunk, Bob's outgoing message is one of four codewords: $\bar{0}$, $\bar{1}$, $\bar{2}$, or $\bar{3}$. He begins the protocol in Phase 0. Once he has learned two possible values of Alice's $x$ he moves onto Phase 1. At some point he transitions to Phase 2, where he remains for the rest of the protocol.

    \begin{enumerate}[align=left, leftmargin=*, label={\bf Phase \arabic*:}]
    \setcounter{enumi}{-1}
    
    \item 
        In Phase 0, Bob's message to Alice is always $\bar{0}$. At the end of the $\frac{km}{\alpha}$'th chunk, where $k \in \bbN$, if fewer than $\frac34 - \frac94\epsilon$ of Alice's last $\frac{m}{\alpha}$ messages have been erased, Bob determines up to two values of $x$ that are consistent with the previous $\frac{m}{\alpha}$ messages. He does this by list decoding the code $C$ as given in Lemma~\ref{lemma:chunk-code}, which is possible by Theorem~\ref{thm:linear}. If there is only one value of $x$, he sets $\hat{x}$ to this unique value. Otherwise, he sets $\hat{x}_0$ and $\hat{x}_1$ to the two values of $x$ and lets $i \in \{ 0, 1 \}^{\log n}$ be an index for which $\hat{x}_0[i] \not= \hat{x}_1[i]$. He then transitions to Phase 1.
        
    \item 
        In Phase 1, Bob's message to Alice is always one of $\bar{0}, \bar{1}, \bar{2}$.
        
        He begins by setting $\mes \gets 1 + i[0]$ and sending $\overline{\mes}$ to Alice. Every message thereafter, let $M \in \{ 0, 1 \}^p$ denote the most recent message he received from Alice and let $j$ denote the number of the chunk minus 1, modulo $\frac{m}{\alpha}$. He determines his behavior depending on which of the following cases $M$ fals under. 
        
        \begin{caseof}
        
        \case {$j \in \BAD$ or $\geq \frac34-\frac32\epsilon$ of the symbols in $m$ are $\bot$.}{
            Bob sends $\overline{\mes}$.
        }
        
        \case {$M$ is $\leq 2$-decoded where there is at most one $b \in \{ 0, 1 \}$ such that one of the decoded elements is of the form $\ECC(\ECC'(\hat{x}_b)[j\alpha, (j+1)\alpha-1], \ind_b)$.}{
            Bob sets $\hat{x} \gets \hat{x}_b$.
        }
        
        \end{caseof}
        
        In all remaining cases, Bob decodes $M$ to two states $\{ \ECC(\ECC'(\hat{x}_0)[j\alpha, (j+1)\alpha-1], \ind_0), \ECC(\ECC'(\hat{x}_1)[j\alpha, (j+1)\alpha-1], \ind_1) \}$ such that $\ECC'(\hat{x}_0)[j\alpha, (j+1)\alpha-1] \not= \ECC'(\hat{x}_1)[j\alpha, (j+1)\alpha-1]$.
        
        \begin{caseof}
        \setcounter{casenum}{3}
        
        \case{For some $b \in \{ 0, 1 \}$, either $|\ind_b| \not\in \{ \nxt-1, \nxt \}$ or $\ind_b$ is not a prefix of $i$.}{
            Bob sets $\hat{x} \gets \hat{x}_{1-b}$.
        }
        
        \case{$|\ind_0| \not= |\ind_1|$.}{
            Bob transitions to Phase 2 with  $\fin \gets 3$ and $\parity \gets (|\ind_1|~\text{mod}~2)$.
        }
        
        \case{$|\ind_0| = |\ind_1| = \nxt-1$.}{
            Bob sends $\overline{\mes}$.
        }
        
        \case{$|\ind_0| = |\ind_1| = \nxt$.}{
            Bob sets $\nxt \gets \nxt+1$ and $\mes \gets (\mes + 1 + i[\nxt] ~\text{mod}~3)$. If $\nxt = \log n$, Bob transitions to Phase 2 with $\fin \gets \mes$ and $\parity \gets \hat{x}_1[i]$. Otherwise, if $\nxt < \log n$, he stays in Phase 1 and sends $\overline{\mes}$.
        }
        
        \end{caseof}
        
    \item 
        When Bob enters Phase 2, he has set values of $\fin \in \{ 0, 1, 2, 3 \}$ and $\parity \in \{ 0, 1 \}$. He sends $\overline{\fin}$ for the rest of the protocol.
        
        At the end of the protocol, let $\beta \in \{ 0, 1 \}$ be the most last bit that Bob received. He sets $\hat{x} \gets \hat{x}_1$ if $\beta = \parity$ and $\hat{x} \gets \hat{x}_0$ otherwise.
    \end{enumerate}
}
\subsection{Analysis}

\begin{claim}
    While Bob has not yet set $\hat{x}$, Alice's value of $\ind$ is a prefix of $i \in \{ 0, 1 \}^{\log n}$. 
\end{claim}

\begin{proof}
    This is clearly true at the beginning of the protocol when $\ind = \emptyset$. Alice only changes $\ind$ when she receives a message $\{ \bar{\gamma} \not= \overline{\mes}$ with $\gamma \in \{ 0, 1, 2 \}$ from Bob such that $\gamma - \mes - 1 \equiv 0,1 \Mod{3}$. Note that this means that she updates $\ind$ only the first time she unambiguously receives a new valued message from Bob. Since while Bob is in Phase 0 he only sends $\bar{0}$, Alice does not update $\ind = \emptyset$ during this time. It remains to show that $\ind$ remains a prefix of $i$ when Bob is in Phase 1 or 2.
    
    This follows from the following three facts: (1) If we consider the sequence $0 = \mes_0 \not= \mes_1 \not= \dots \not= \mes_k \in \{ 0, 1, 2 \}$ (excluding the final message $\bar{3}$ that Bob may send in Phase 2) of different messages that Bob sends throughout the protocol, then $(\mes_{\iota+1} - \mes_\iota - 1 ~\text{mod}~3) = \ind_\iota$. (2) If we consider the sequence of values $\gamma_1 \not= \gamma_2 \not = \dots \not= \gamma_\ell \in \{ 0, 1, 2 \}$ that Alice unambiguously decodes Bob's messages to while Bob is in Phase 1 and 2, discarding contiguous repeats and any final $\bar{3}$, then $\gamma_\iota = \mes_\iota$. This is true because Bob sends the same value of $\mes$ until he 2-decodes Alice's message and sees that $|\ind_0| = |\ind_1| = \nxt$, and since Alice's real value of $\ind$ is either $\ind_0$ or $\ind_1$, both of which have length $\nxt$, Alice must've unambiguously decoded one of his messages to $\mes$. (3) Alice updates $\ind$ $\ell$ times, each time appending $(\gamma_{\iota+1} - \gamma_\iota - 1) ~\text{mod}~3$ to $\ind$, for $\iota = 0, \dots, \ell - 1$, where $\gamma_0 = 0$. 
\end{proof}

\begin{theorem}
    Protocol~\ref{prot:6/11} is resilient to a $\frac6{11}-O(\epsilon)$ fraction of erasures. For an input of size $n$, the total communication is $O_\epsilon(n)$. Alice and Bob run in $\poly_\epsilon(n)$ time.
\end{theorem}

\begin{proof}
    We first analyze the communication complexity. There are $T = \frac{m}{\epsilon \alpha}$ chunks, each of which has $p + \frac{3p}{8} = \frac{11p}{8}$ bits sent. We have that $m = O_\epsilon(n)$ and $p = O_\epsilon(\alpha + \log n)$ by Theorem~\ref{thm:ECC}. Since $\alpha = \theta(\log n)$, it holds that $p = O_\epsilon(\alpha)$, so that the total number of bits sent is 
    \[
        \frac{m}{\epsilon \alpha} \cdot \frac{11p}{8} 
        = \frac{O_\epsilon(n)}{\alpha} \cdot O_\epsilon(\alpha) 
        = O_\epsilon(n).
    \] 
    
    Now, we show erasure resilience. Suppose for the sake of contradiction that Bob outputs an incorrect value of $x$. We will show that the adversary must've corrupted more than $\frac6{11} - O(\epsilon)$ of the communicated bits.
    
    First, we claim that if Alice ever uniquely decodes Bob's message in chunk $R$ once he is in Phase 2, and Bob hears at least one bit from Alice in a chunk after $R$, then Bob will output the correct value of Alice's input $x$. To see this, note that if Bob enters Phase 2 with $\fin = 3$, then this means that $\nxt \le \log n - 1$ and Bob just 2-decoded Alice's message such that w.l.o.g. $|\ind_0| = \nxt-1$ and $|\ind_1| = \nxt$ (it's not possible for $|\ind_0| = \log n - 1$ and $|\ind_1| = \log n$ since that requires that at some previous point $|\ind_0| = |\ind_1| = \log n - 1$, at which point Bob transitions to Phase 2 with $\fin = \mes \in \{ 0, 1, 2 \}$). Then, when Alice unambiguously hears a $\bar{3}$, she sends the bit $|\ind|~\text{mod}~2$ for the rest of the protocol, which allows Bob to determine whether Alice had $\ind_0$ or $\ind_1$ and thus whether her input were $\hat{x}_0$ or $\hat{x}_1$. If Bob enters Phase 3 with $\fin \in \{ 0, 1, 2 \}$, it must be the case that $\nxt = \log n$ and that in the previously 2-decoded message, $|\ind_0| = |\ind_1| = \log n - 1$. Then, if Alice receives Bob's Phase 3 message, she learns the last bit of $i$ and switches to sending $x[i]$ for the rest of the protocol. Since $\hat{x}_0[i] \not= \hat{x}_1[i]$, Bob can use this to distinguish between $\hat{x}_0$ and $\hat{x}_1$.
    
    Let $R$ be the first chunk in which Alice uniquely decodes Bob's message while he's in Phase 2, and if such a chunk does not exist then let $R=T$. The argument above implies that if Bob outputs the incorrect value of $x$, it must be the case that none of Alice's messages after chunk $R$ got through to Bob. Also let $U$ be the last chunk in which Bob is in Phase 1, so that $U = \frac{km}{\alpha} < R$ for some $k \in \bbN$.
    
    Since Bob did not transition to Phase 1 earlier, there must've been $\ge \frac34 - \frac94$ erasures in Alice's first $\frac{(k-1)m}{\alpha}$ messages. In Alice's $\frac{(k-1)m}{\alpha} + 1$'th to $\frac{km}{\alpha}$'th messages, since there were more than one consistent value of $x$, there must've been at least $\frac12 - \frac32 \epsilon$ erasures. This mean that within the first $U$ chunks, there are at least 
    \begin{align*}
        \left( \frac34 - \frac94 \epsilon \right)  \cdot p \cdot \left( U - \frac{m}{\alpha} \right) + \left( \frac12 - \frac32 \epsilon \right) \cdot p \cdot \frac{m}{\alpha}
        &= \left( \frac34 - \frac94 \epsilon \right) \cdot p \cdot U - \left( \frac14 - \frac34 \epsilon \right) \cdot p \cdot \frac{m}{\alpha} \\
        &= \left( \frac34 - \frac94 \epsilon \right) \cdot p \cdot U - \left( \frac14 - \frac34 \epsilon \right) \cdot p \cdot \epsilon T
    \end{align*}
    erasures.
    
    In the ($U$+1)'th to $R$'th chunks, Alice sends $\ECC(\ECC'(x)[j\alpha, (j+1)\alpha-1], \ind)$. Since we assumed that Bob outputs $\hat{x} \not= x$, it must be the case that none of these messages where $j \not\in \BAD$ can be uniquely decoded, and in particular at least $\frac12 - \epsilon$ of each of Alice's messages where $j \not\in \BAD$ must be erased, otherwise Bob uniquely decodes Alice's message and sets $\hat{x}$ correctly. Let $S$ be the number of these $R-U$ chunks in which at least $\frac34 - \frac32 \epsilon$ of Alice's bits are erased. In the other $R - U - S$ chunks, either $j \in \BAD$ or between $\frac12-\epsilon$ and $\frac34-\frac32\epsilon$ of Alice's message is erased. There are at most $\frac{\epsilon m}{\alpha} \cdot \left\lceil \frac{R-U}{m/\alpha} \right\rceil \le \epsilon (R - U) + \frac{\epsilon m}{\alpha} = \epsilon(R - U) + \epsilon T$ chunks among these with $j \in \BAD$. Then there are $\ge (1-\epsilon) (R - U) - \epsilon T - S$ chunks such that between $\frac12 - \epsilon$ and $\frac34 - \frac32 \epsilon$ of Alice's message is erased. We argue that in at most $\log n$ of these chunks, Bob's messages to Alice have a unique decoding. This is the case since whenever Alice uniquely decodes Bob's message, she appends one bit to $\ind$, or Bob's message was a $\bar{3}$ (in which case $|\ind| < \log n$), and $|\ind| \le \log n$. In the other $\ge (1-2\epsilon)(R - U) - S - \log n$ chunks, Bob's message to Alice is at least $\frac23$ corrupted. This gives a total number of erased bits of
    \begin{align*}
        & \ge \left( \frac34 - \frac94 \epsilon \right) \cdot p \cdot U - \left( \frac14 - \frac34 \epsilon \right) \cdot p \cdot \epsilon T
        + \left( \frac34 - \frac32\epsilon \right) \cdot p \cdot S + p \cdot (T - R) \\
        & + \left( \frac12 - \epsilon \right) \cdot p \cdot \left( (1-\epsilon)(R - U) - \epsilon T - S \right) + \frac23 \cdot \frac{3p}{8} \cdot \left( (1 - \epsilon)(R - U) - \epsilon T - S - \log n \right) \\
        &\ge \left( 1 - \epsilon \right) p \cdot T - \left( \frac14 + \frac74\epsilon \right) \cdot p \cdot R - \left( \frac12 \epsilon + \epsilon^2 \right) p \cdot U - \frac12 \epsilon p \cdot S - \frac14 p \cdot \log n \\
        &\ge \left( 1 - \epsilon \right) p \cdot T - \left( \frac14 + \frac74\epsilon \right) \cdot p \cdot T - \left( \frac12 \epsilon + \epsilon^2 \right) p \cdot T - \frac12 \epsilon p \cdot T - \frac14 p \cdot \epsilon T \\
        &= \left( \frac34 - O(\epsilon) \right) p \cdot T,
    \end{align*}
    where we used that $R, U, S \le T$ and $\log n = o(T)$.
    
    In the whole protocol, there are $\frac{11}{8} p \cdot T$ bits communicated, so the fraction of bits the adversary must've erased is $\frac{\frac34 - O(\epsilon)}{11/8} = \frac{6}{11} - O(\epsilon)$.
\end{proof}

\bibliographystyle{alpha}
\bibliography{refs}

\begin{thebibliography}{GKZ21}

\bibitem[GKZ21]{GuptaTZ21}
Meghal Gupta, Yael~Tauman Kalai, and Rachel~Yun Zhang.
\newblock Interactive error correcting codes over binary erasure channels
  resilient to $>\frac12$ adversarial corruption, 2021.

\bibitem[GS00]{GuruswamiS00}
Venkatesan Guruswami and Madhu Sudan.
\newblock List decoding algorithms for certain concatenated codes.
\newblock In {\em Proceedings of the Thirty-Second Annual ACM Symposium on
  Theory of Computing}, STOC '00, page 181–190, New York, NY, USA, 2000.
  Association for Computing Machinery.

\bibitem[Gur03]{Guruswami03}
V.~Guruswami.
\newblock List decoding from erasures: bounds and code constructions.
\newblock {\em IEEE Transactions on Information Theory}, 49(11):2826--2833,
  2003.

\bibitem[Ham50]{Hamming50}
R.~W. Hamming.
\newblock Error detecting and error correcting codes.
\newblock {\em The Bell System Technical Journal}, 29(2):147--160, 1950.

\bibitem[Sha48]{Shannon48}
C.~E. Shannon.
\newblock A mathematical theory of communication.
\newblock {\em The Bell System Technical Journal}, 27(3):379--423, 1948.

\end{thebibliography}

\end{document}